\documentclass[reqno]{amsart}
\usepackage[english]{babel}
\usepackage[latin1]{inputenc}
\usepackage{amsfonts}
\usepackage{amssymb}
\usepackage{amsmath}
\usepackage{amsthm}
\usepackage{graphicx}
\usepackage{epsfig}
\usepackage{fancyhdr}
\usepackage[all]{xy}
\usepackage{xcolor}
\usepackage{hyperref}
\usepackage{amsmath,amssymb}
\usepackage[bbgreekl]{mathbbol}

\numberwithin{equation}{section}

\newcommand{\C}{\mathbb{C}}

\newtheorem{theorem}{Theorem}[section]
\newtheorem{lemma}[theorem]{Lemma}

\newtheorem{defi}[theorem]{Definition}

\newcommand{\ip}[2]
      {\ensuremath{\langle #1,#2 \rangle}}

\usepackage{amsmath}
\title[Self-Adjoint realizations of  2d-Dimensional Canonical Systems]{ Self-Adjoint realizations of 2d-Dimensional Canonical Systems and applications}
\author{Keshav Acharya, Andrei Ludu}
\date{March 2025}

\begin{document}

\maketitle

\begin{abstract}
This paper studies linear relations and their self-adjoint realizations arising from $2d$-dimensional canonical systems, with a focus on how the symplectic structure interacts with boundary conditions. Understanding this interplay allows us to define self-adjoint realizations, which are crucial for analyzing the spectral properties of these systems. We prove that for each pair of Lagrangian boundary matrices $\Theta$ and $B$ satisfying appropriate orthonormality conditions, the restricted relation $T_{\Theta,B}$ is self-adjoint. Our approach relies on the symplectic geometry of boundary spaces and the isotropic structure of Lagrangian subspaces. We also discuss extensions to semi-infinite intervals using asymptotic boundary conditions. 

In the second part of the paper, we show how this framework applies to spectral problems from partial differential equations, including the stability of traveling waves and the linearization of the focusing nonlinear Schr\"odinger equation around soliton profiles. In particular, the self-adjoint structure with respect to the $H$-weighted inner product ensures the spectral properties needed for stability analysis using Evans function and transfer matrix methods. Applications to integrable systems, such as the stability of NLS bright solitons, are also presented.\\
\vspace{.2in}

\noindent \textbf{ Keywords:}   Canonical system, Linear Relations, Self-afjoint relations, Spectrum, Lagrangian subspace.  

\end{abstract}

\section{ Introduction }

 The theory of Canonical systems serves as a foundational tools in spectral theory, integrable systems, and mathematical physics. Their classical formulation provides a unifying framework for diverse equations such as Schr\"odinger, Dirac, and Jacobi systems, connecting operator theory and differential equations through a common symplectic structure. These systmes in two dimensions, initially studied by L. de Brange in \cite {deBranges1968}, have been extensively studied, see \cite{KRA, HSW, HSW1, HSW2, math1, Sakhnovich2010} and references therein. Extending these systems  to  $2d$-dimensions is both a natural and profound generalization that introduces deeper analytical challenges. In higher dimensions, these systems capture multi-channel interactions, internal symmetries, and couplings between distinct physical or mathematical modes-features. Consequently, they play a central role in  multi-channel quantum scattering,  elasticity theory,  matrix-valued Sturm-Liouville problems, and  systems governed by transfer matrices, see \cite{ Sakhnovich2010}.

A general $2d$-dimensional canonical system is expressed as
\begin{equation} \label{ca}
J \frac{dy}{dx} = zH(x)y(x), \quad z \in \mathbb{C},
\end{equation}
where $y(x) \in \mathbb{C}^{2d}$ and
$
J =
\begin{pmatrix}
0 & I_d \\
-I_d & 0
\end{pmatrix}
$
is the standard symplectic matrix of size $2d \times 2d$. The matrix $H(x)$ is a locally integrable, positive semi-definite, Hermitian $2d \times 2d$ matrix that serves as the \emph{Hamiltonian} of the system. Moreover, we require that there is no non-empty interval $I \subset [0, \infty)$ on which $ H $ vanishes almost everywhre. The complex number $z$ is a spectral parameter. In general, $H(x)$ may not be invertible at all points, reflecting physical degeneracies or symmetries that constrain the system. Such degeneracies are typical in matrix-valued and coupled problems and lead to nontrivial analytical complications, particularly regarding operator domains and self-adjointness.  These systems in 2d dimensions have also been studying, see some recent work in \cite{AcharyaLudu2021, AcharyaLudu2025}  

The study of solutions $u(x, z)$ to the canonical system (1.1) often focuses on those that are \emph{$H$-integrable}, satisfying
\begin{equation}
\int_{0}^{\infty} u^*(x)H(x)u(x)\,dx < \infty,
\end{equation}
 where $ ^*$ denote the conjugate transpose. The class of all such functions forms a Hilbert space $L^2(H, [0, \infty))$, defined by
\begin{equation}
L^2(H, [0, \infty)) =
\left\{
f :  \int_{0}^{\infty} f^*(x)H(x)f(x)\,dt < \infty
\right\}.
\end{equation}
Since $H(x)$ may be singular, two functions $f$ and $g$ are identified whenever $Hf = Hg$. Under this equivalence relation, one defines the inner product
\begin{equation}
\langle f, g \rangle =
\int_{0}^{\infty} f^*(x)H(x)g(x)\,dt,
\end{equation}
which is independent of the representative of the equivalence class and makes $L^2(H, [0, \infty))$ into a Hilbert space. This functional framework allows one to treat canonical systems as generalized differential relations rather than operators, a distinction that becomes essential when $H(x)$ fails to be strictly positive definite.

In this context, the differential expression
\[
\tau y = J^{-1}zH(x)y
\]
is not necessarily densely defined in $L^2(H)$. Instead, it induces a \emph{linear relation}, i.e., a multivalued closed subspace of $L^2(H) \times L^2(H)$, which generalizes the notion of a differential operator. The analysis of such relations provides a robust pathway to handle degeneracies while retaining self-adjointness and spectral interpretation \cite{Arens1961, DerkachMalamud1995}.  

A central question in this framework concerns the characterization of self-adjoint extensions of the symmetric linear relation induced by \eqref{ca}. Self-adjointness guarantees real spectra and unitary dynamics, both fundamental in quantum and wave systems.The motivation for studying higher-dimensional canonical systems extends beyond pure mathematics because of profound applications. \\

  Therefore, in this project we aim to systematically study the self-adjoint realizations of these systems and present some important applications, offering new insights into both the theoretical and applied aspects of these systems in higher dimensions. In section 2,  the fundamental theory of linear relations and their self-adjoint extensions with the appropriate boundary conditions are presented. Then in section 3,  some applications are presented, illustrating how these self-adjoint extensions can be used to analyze spectral properties, wave propagation, and other phenomena in physical and engineering systems.

\section{Preliminaries and Results}

 Define  a maximal relation $\mathcal T$ induced by a $2d$ dimensional canonical systems   \eqref{ca}, as follows.
\begin{align}
 \label{rel} \mathcal T = \{ (u,v) \in (L^2(H, I))^2 :  Ju' = H v \ \ \text{ a. e.} \text{ on } I \}. \end{align} Here $ I = [0, N]$ or $ [0, \infty).$
Also, $Ju'= Hv$ means that  $u$ has a absolutely continuous representative that satisfy  $Ju'= Hv$  almost everywhere on $I$.
The right-hand side does not depend on the choice of representative but the left side may depend. Therefore, a vector $u$ under  $\mathcal T$ can have several images. The set $D(\mathcal T) = \{f: (f,g) \in \mathcal T\}$ and $R(\mathcal T)= \{g: (f,g) \in \mathcal T\}$ are respectively the domain and the range of $\mathcal T.$ These relations for 2 dimensional canonical systems are well explained in \cite{KRA, math1}.

\begin{defi} Let $ \mathcal H$ be a Hilbert space and $\mathcal T = \{ (f,g): f, g \in \mathcal H\} \subset \mathcal H \oplus \mathcal H$ be a linear relation on $\mathcal H$ we define  the inverse of $\mathcal T$ is the relation $$ \mathcal T^{-1} = \{(g, f): (f,g) \in \mathcal T \}$$ and the adjoint of $\mathcal T$ is the relation $$ \mathcal T^* = \{ (h,k): \ip{h}{g} = \ip{k}{f} \text{  for all  }  (f,g) \in \mathcal T\}.$$ A relation $\mathcal T$ is symmetric if $\mathcal T^* \subseteq  \mathcal T$
and self-adjoint if $\mathcal T^* = \mathcal T .$
\end{defi}

The spectral analysis of system \eqref{ca} is carried out via the self-adjoint relation naturally induced by the system. In this context, Green's identity plays a central role, as it provides the fundamental formula for the self-adjoint relation. We explicitly establish this identity for the relation. We establish the identity for the relation $ \mathcal T$ defined in \eqref{rel} on the bounded interval $I =[0, N]$.  

\begin{lemma} \label{lemma 1}For any  $(u,v), (f, g) \in \mathcal T,  $ we have \begin{equation} \label{gi}\ip{u}{g}-\ip{v}{f}= u^*(N)Jf(N) -u^*(0)Jf(0).\end{equation} \end{lemma}
\begin{proof} For any  $(u,v), (f, g) \in \mathcal T,  $ \begin{align*}  \ip{u}{g}-\ip{v}{f} & = \int_0^N u^*(x)H(x)g(x)dx -\int_0^{N} v^*(x)H(x)f(x)dx \\ & =  \int_0^N \Big(u^*(x)Jf'(x) - (Ju'(x))^* f(x)  \Big) \ dt 
	\\& =  \int_0^N \Big(u^*(x)Jf'(x) + u'(x))^* J f(x)  \Big)  	\\ & =  \int_0^N  \frac{d}{dt}\Big(u^*(x)Jf(x)) \Big)  \\ & = u^*(N)Jf(N) -u^*(0)Jf(0) .\end{align*} \end{proof}

In order to obtain a self-adjoint extension of $\mathcal T$ we need the right hand side of  \eqref{gi} to be zero separately.  That is  $u^*(N)Jf(N) = 0 $ and  $u^*(0)Jf(0) = 0.$ This is obtained by imposing suitable boundary conditions for all $f \in D(\mathcal T)$ at 0 and N.  \\

First, observe that the expressions $u^*(N)Jf(N)$ or $ u^*(0)Jf(0)$ induce a bilinear form in the vector space $\C^{2d}.$ So equip $\C^{2d}$ with a bilinear form given by $b(p,q) = q^*Jp,$ and an inner product $\langle p, q \rangle = q^*p $ for all $p, q \in \C^{2d}.$ This bilinear form is nondegenerate and skew-symmetric. Therefore $(\C^{2d}, b)$ is a symplectic vector space. Recall that a subspace $W$ of $\C^{2d}$ is called isotropic if $b(w_1, w_2) =  0$ for all $w_1, w_2 \in W .$ 
 This can also be expressed as $ W \subseteq W^{\perp_b} ,$  where $ W^{\perp_b} = \{ v \in C^{2d} : b(v, w) = 0, \, \text{for all } w \in W\}$.
 
 A maximal Lagrangian subspace of a symplectic vector space is a largest possible isotropic subspace. Equivalently, they satisfy $ W=W^{\perp_b}.$ If the symplectic space is $ 2d $ dimensional, then the Lagrangian subspace is $ d $ dimensional. These Lagraian subspaces are explained and applied for the discrete canonical systems in \cite{FR}.

 Suppose $ \mathcal L $ is the space of all $ \Theta =  \begin{pmatrix} \theta_1 \\ \theta_2 \end{pmatrix} \in C^{2d\times d}$,  $ \theta_1, \theta_2 \in\C^ {d\times d}  $ satisfying \begin{equation} \label{bc} \theta_1 \theta_1^* + \theta_2 \theta_2^* = 1, \,\,\ \ \  \theta_1 \theta_2^* - \theta_2 \theta_1^* = 0    . \end{equation}

 For fixed $ \Theta ,\   B \in \mathcal L $  we impose the  following boundary conditions in \eqref{ca}
\begin{equation} \label{bc1} (\theta_1, \theta _2 ) u(0) = 0, \ \ \  (\beta_1, \beta_2)u(N) = 0 .\end{equation} and write the linear relation; which is a minimal relation;
\begin{align} \label{bc1} \nonumber  \mathcal T_{\Theta, B} =\{ (u,v) \in (L^2(H, [0,N]))^2 : Ju'= Hv, a. e. , \\  (\theta_1, \theta _2 ) u(0) = 0, \ \   (\beta_1, \beta_2)u(N) = 0  \}.\end{align}

\begin{theorem} \label{thm1} For each $ \Theta ,\  B  \in \mathcal L $ the  relation $\mathcal T_{\Theta, B} $ is self-adjoint.\end{theorem} 
In order to prove the theorem, we proceed as follows. Let $\Theta = \begin{pmatrix}\theta_1, \theta_2\end{pmatrix}$. Observe that \[ \Theta J \Theta^* =   (\theta_1, \theta_2) \begin{pmatrix} 0 & I \\ -I & 0 \end{pmatrix}
\begin{pmatrix} \theta_1^*   \\ \theta_2^* \end{pmatrix} = \theta_1 \theta_2^* - \theta_2 \theta_1^* = 0. \] 

Consider a map $\Theta^*: \C^{d} \rightarrow \C^{2d}$   such that $\Theta^*u = \begin{pmatrix} \theta_1^* u \\ \theta_2^* u \end{pmatrix}.$ Then the image $\operatorname{Im } \Theta^* = \{w: w= \Theta^*v, \, \text{for some } v \in \C^d\}$ of the map $\Theta^*$ is a subspace of $\C^{2d} ,$ and we have the following lemma.

\begin{lemma} \label{lemma 1.3} The subspace $ W = \operatorname{Im } \Theta^*$ is a Lagrangian subspace of the symplectic space $\C^{2d}$ and $ W^\perp=JW  $. \end{lemma}

\begin{proof} Let $w_1=\Theta^*u $, $w_2=\Theta^*v$ with $u, v\in \C^d$. Then
\[
w_1^* J w_2 
= (\Theta^*u)^* J (\Theta^*v)
= u^* (\Theta J \Theta^*) v
= 0 .\] Since this holds for all $u,v\in\C^d$, the subspace $W$ is
indeed isotropic. Moreover, since $\Theta \Theta^* = I$, $\dim W =d.$ It follows that $W$ is a Lagrangian subspace.

Next, let $ Jw \in JW $ then for any $w' \in W , \,$  $\langle Jw, w' \rangle = w'^*Jw = 0 .$ So $JW \subseteq W^{\perp}.$ Also,  since $ \dim (J W) = \dim (W) = d  $  and $\dim W^{\perp} = 2d -d =d $,  $ JW = W^{\perp}   $ \end{proof}

\begin{lemma}\label{lemma2.5}
Let $p=\binom{p_1}{p_2},\,q=\binom{q_1}{q_2}\in \C^{2d}$. 
Suppose $\Theta \in \mathcal L$. If $\Theta p=0$ and $\Theta q=0$, 
then $q^*Jp=0.$ Moreover, if $q^*Jp=0$ for all $q$ with $\Theta q=0$ then $\Theta p =0$.
\end{lemma}

\begin{proof}
In lemma \ref{lemma 1.3} we showed that the
subspace $W = \operatorname{Im } \Theta^*$  is Lagrangian. Consider  the subspace \[ \ker \Theta := \{\, w \in \mathbb C^{2d} : (\theta_1,\theta_2)w = 0 \,\}. \]

Since $\ker \Theta=(\operatorname{Im}\Theta^*)^\perp=W^\perp =JW$.  If $p,q\in\ker \Theta$, then $p=Jw_p$, $q=Jw_q$ with $w_p,w_q\in W$. Thus
\[ q^*Jp=(Jw_q)^*J(Jw_p)=w_q^*Jw_p.\]
Since $W$ is isotropic, $w_q^*Jw_p=0$. Therefore, $q^*Jp=0$.

 Next, suppose $ q^*Jp = 0, $ for all $ q \in \ker \Theta .$ Then $$Jp \in ( \ker \Theta)^{\perp} = (JW)^{\perp} = -JW^{\perp}.$$ It follows that $p \in W^{\perp}= \ker \Theta.$ That is $ \Theta p =0 .$
 \end{proof}
  
 If \((f,g)\in \mathcal T_{\Theta, B}^*\) then \((\theta_1,\theta_2)f(0)=0\) and \((\beta_1,\beta_2)f(N)=0\).

\begin{proof}[Proof of Theorem \ref{thm1}]

  First, we show that $\mathcal T_{\Theta, B} $ is symmetric, that is, $ \mathcal T_{\Theta, B} \subseteq \mathcal T_{\Theta, B}^* .$ Let $(h,k) \in \mathcal T_{\Theta, B} $. By  Using the Green's identity \eqref{lemma 1}, for every $(f,g)\in \mathcal T_{\Theta, B}$ we have
\begin{equation} \label{gieq}
    \langle g,h \rangle - \langle f,k \rangle.
 = f(N)^* J h(N) - f(0)^* J h(0) .\end{equation}
 
 By lemma \ref{lemma2.5}, $f(N)^* J h(N) $ and $f(0)^* J h(0) $ are separately zero.  Therefore \begin{equation*}
    \langle g,h \rangle - \langle f,k \rangle
 = 0\end{equation*} and hence   $(h,k) \in \mathcal T_{\Theta, B}^* .$ So, $\mathcal T_{\Theta, B}$ is symmetric.
Next we show that $\mathcal T_{\theta, \beta}^* \subseteq \mathcal T_{\theta, \beta} .$ Suppose $(h,k) \in \mathcal T_{\Theta, B}^* .$ Then by definition \begin{equation} \label{eqgi1}
  \langle g,h \rangle = \langle f,k \rangle, \quad \forall \, (f,g) \in \mathcal T_{\Theta, B}.\end{equation} Again by the Green's identity we have, 
\begin{equation} \label{eqgi2} \langle g,h \rangle - \langle f,k \rangle = f(N)^* J h(N) - f(0)^* J h(0) . \end{equation} Then by equations \eqref{eqgi1} and \eqref{eqgi2}, we get \begin{equation} \label{eqgi}  f(N)^* J h(N) - f(0)^* J h(0)  =0 .\end{equation} Since \eqref{eqgi1} is true for all $(f,g) \in \mathcal T_{\Theta, B},$ we may chose $f$ with $f(N) = 0 $ so the equation \eqref{eqgi} reduces to $ f(0)^* J h(0)  =0 .$ By lemma \ref{lemma2.5} $ \Theta h(0)= 0.$ That is $h$ satisfies the boundary condition at $0.$ Similarly, one can see that $h$ satisfies boundary condition at $N.$ Next we show that $Jh'= Hk$ a. e.\ . Since the equation \eqref{eqgi1} is true for all $(f,g) \in \mathcal T_{\Theta, B} ,$ we may chose $f$ with compact support. By equation \eqref{eqgi1} and the definition of inner product $\langle x,y\rangle:=\int_0^N x(t)^* H(t) y(t)\,dt$ we obtain from equation \eqref{eqgi1} that
 
\begin{align*}0 & =\int_0^N k^*(t)H(t)f(t)\,dt - \int_0^N h^*(t)Hg\,dt\\ & = \int_0^N k^*(t)H(t)f(t)\,dt - \int_0^N h^*(t)Jf'(t)\,dt
\end{align*}
Integrate the second term by parts. Since $f$ is absolutely continuous on $[0,N]$,
\[
\int_0^N h^* J f' = \big[h^*(t)J f(t)\big]_{t=0}^{t=N} - \int_0^N h'^*(t) J f(t)\,dt.
\]
Hence
\[
\int_0^N \big(k^*H + h'^* J\big)f \,dt
= \big[h^*(t)J f(t)\big]_{0}^{N}.
\]

Now choose $(f,g)\in T_{\Theta,B}$ with $f$ compactly supported in $(0,N)$ (this is possible by standard cutoff arguments), so the boundary term on the right vanishes. For such $f$ we obtain
\[
\int_0^N \big(k^*H + h'^* J\big)f \,dt =0.
\]
Because these equalities hold for all compactly supported $u$, the integrand must vanish almost everywhere:
\[
k^*(t)H(t) + h'^*(t) J = 0 \quad\text{a.e. }t\in(0,N).
\]
Using $J^*=-J$ (so $h'^* J = - (J h')^*$) we rewrite this as
\[
k^*H - (J h')^* = 0 \quad\text{a.e.}
\]
Taking adjoints gives the desired differential relation
\[
H k = J h' \quad\text{a.e. on }(0,N) 
.\] This concludes that $(h,k) \in \mathcal T_{\Theta, B}.$ Therefore $ \mathcal T_{\Theta, B} = \mathcal T_{\Theta, B}^*.$

\end{proof}

Next we extend the Theorem \ref{thm1} on the  half line $[0, \infty).$ we proceed by extending  the Green's identity: 
  For any  $(u,v), (f, g) \in \mathcal T,  $ we have \begin{equation} \label{gi}\ip{u}{g}-\ip{v}{f}= \lim_{N\rightarrow \infty} u^*(N)Jf(N) -u^*(0)Jf(0).\end{equation} 

 In order to obtain the self-adjoint extension of $\mathcal T$ the expressions on the right side of this lemma must be 0 separately. Consider a relation
 \begin{equation} \mathcal T^{\Theta, p} = \{ (f,g) \in \mathcal T:  (\theta_1, \theta_2)f(0) = 0, \quad \lim_{x\rightarrow \infty } f(x)^*JP(x) = 0 \} ,\end{equation} where $p(x) \in D(\mathcal T)$ has compact support. It can similarly be shown that the relation $ \mathcal T^{\Theta, p}  $ is self-adjoint.

\section{Physics Applications}
\label{sec.phys}

The classical canonical system  (linear Hamiltonian system) in Eq. (\ref{ca}) occurs in many concrete physics areas where this form (or some of its immediate close variants) appears \cite{math1,appl1,applqm,appl2}. This equation is the prototype canonical or Hamiltonian first-order system that appears widely in physics: small-oscillation Hamiltonian mechanics \cite{appl1}, reduction of wave/PDE spectral problems, waveguides and optics, electrical transmission lines \cite{appl2}, elasticity, and integrable scattering problems \cite{applqm}. The positive semidefinite Hermitian matrix $H$ ties to local energy, and the symplectic matrix $J$  enforces conservation laws, both being central to analysis, numerics, and physical interpretation. These problems are also related to standard tools of mathematical-physics such as fundamental (time-ordered) exponentials, Riccati reductions, Floquet theory, symplectic integrators and transfer/Evans-function methods.

A direct physical application in spatial dynamics for the above mathematical model is given, for example, by spectral problems from PDEs related to traveling waves. When reducing second-order ODEs (PDEs) (e.g. 1D Schr\"{o}dinger, elastic beam, or linearized PDE about a traveling wave) to first order in space, one obtains a $2n \times 2n$ Hamiltonian system where the eigenvalue (spectral boundary-value problems) take the form 
$$
Jy'=(H-\lambda M)y.
$$
The analysis presented above can help the study of the spectrum, to locate eigenvalues (stability), or compute Evans functions or transfer matrices. Many spectral problems arising from partial differential equations can be cast as $2d$-dimensional canonical systems through appropriate reformulation \cite{math1}. Consider the general second-order spatial operator in one dimension occurring in Sturm-Liouville problems
\begin{equation}
\mathcal{L}u = -\frac{d}{dx}\left(p(x)\frac{du}{dx}\right) + q(x)u = \lambda \rho(x)u, \quad x \in [0,N]
\end{equation}
where $p(x) > 0$, $\rho(x) > 0$, and $q(x) \geq 0$ are given coefficient functions. This appears, for example, in Schr\"odinger equations (with $p=1$, $\rho=1$, $q=V(x)$ the potential), elastic beam vibrations, and linearizations around traveling wave solutions. To transform this into canonical form, we introduce the phase-space variables:
\begin{equation}
y = \begin{pmatrix} u \\ p(x)\frac{du}{dx} \end{pmatrix}
\end{equation}
The spectral problem then becomes
\begin{equation}
J\frac{dy}{dx} = \lambda H(x)y - V(x)y
\end{equation}
where 
$$
J = \begin{pmatrix} 0 & 1 \\ -1 & 0 \end{pmatrix}, \ \ H(x) = \begin{pmatrix} \rho(x) & 0 \\ 0 & 0 \end{pmatrix}, \hbox{ and } V(x) = \begin{pmatrix} 0 & 0 \\ q(x) & 0 \end{pmatrix}.
$$
This can be rewritten in the standard form (1.1) by defining $z = \lambda$ and $\tilde{H}(x) = H(x) - \frac{1}{\lambda}V(x)$ for $\lambda \neq 0$. More generally, for systems of PDEs or higher-order equations, we obtain $2d$-dimensional canonical systems with $d > 1$.

A particularly important application arises in the stability analysis of traveling wave solutions \cite{wave}. Consider a nonlinear PDE describing gravity  long  crested dispersive-dissipative surface waves of the form:
\begin{equation}
u_t = u_{xx} + f(u,u_x)
\end{equation}
admitting a traveling wave solution $u(x,t) = \phi(x-ct)$ where $\phi(\xi)$ satisfies an ODE. The linearized stability problem around this wave leads to:
\begin{equation}
w_t = \mathcal{L}_\phi w := w_{xx} + a(\xi)w_x + b(\xi)w
\end{equation}
where $a(\xi)$ and $b(\xi)$ depend on $\phi$ and its derivatives. Seeking normal mode solutions $w(\xi,t) = e^{\lambda t}v(\xi)$ yields the spectral problem:
\begin{equation}
\mathcal{L}_\phi v = \lambda v
\end{equation}
Converting to first-order form with $y = (v, v')^T$ gives:
\begin{equation}
J\frac{dy}{d\xi} = (H(\xi) - \lambda M)y
\end{equation}
where $M = \begin{pmatrix} 1 & 0 \\ 0 & 0 \end{pmatrix}$ and $H(\xi) = \begin{pmatrix} b(\xi) & 0 \\ -a(\xi) & 1 \end{pmatrix}$.
The eigenvalues $\lambda$ with $\text{Re}(\lambda) > 0$ indicate unstable modes. The essential spectrum can be determined by analyzing the asymptotic behavior as $|\xi| \to \infty$, while point spectrum (isolated eigenvalues) requires solving the boundary value problem on $[0,N]$ with appropriate conditions.


As another application of our mathematical model, we mention that the self-adjoint realization established in Theorem 2.3 provides crucial information for spectral analysis. On one hand, we mention the reality of spectrum under suitable conditions. When $H(x)$ and the boundary conditions at 0 and $N$ satisfy appropriate reality/symmetry conditions, the relation $\mathcal T_{\Theta, B}$ can be shown to have real spectrum. This is essential for stability analysis---complex eigenvalues with positive real part indicate instability. Another topic is related to the orthogonality of eigenfunctions. For distinct eigenvalues $\lambda_m \neq \lambda_n$, the corresponding eigenfunctions $(y_m, z\lambda_m H(x)y_m)$ and $(y_n, z\lambda_n H(x)y_n)$ satisfy
\begin{equation}
\langle \lambda_m H y_m, y_n \rangle = \langle y_m, \lambda_n H y_n \rangle
\end{equation}
Using Green's identity Eq. (2.2) with appropriate boundary conditions, this yields orthogonality in the $H$-weighted inner product. This is fundamental for mode decomposition and completeness results. The model also relates to the variational characterization, because the self-adjointness allows variational formulations for eigenvalues. The smallest eigenvalue can be characterized as:
\begin{equation}
\lambda_1 = \inf_{y \in D(\mathcal T_{\Theta, B})} \frac{\langle y, Jy' \rangle}{\langle y, Hy \rangle}
\end{equation}
providing stability criteria: if $\lambda_1 < 0$, the system is unstable.


The model formulated in our paper can be extended even to spectral problems on unbounded domains (semi-infinite or infinite intervals), using the Evans function $\mathcal{E}(\lambda)$  as a tool for locating point spectrum \cite{eva}. For our canonical system Eq. (\ref{ca}):
\begin{equation}
J\frac{dy}{dx} = (zH(x) - \lambda M)y
\end{equation}
the Evans function is constructed from fundamental matrix solutions satisfying asymptotic boundary conditions, while the zeros of $\mathcal{E}(\lambda)$ correspond to eigenvalues. The symplectic structure of our canonical system ensures that $\mathcal{E}(\lambda)$ is analytic in appropriate regions of the complex $\lambda$-plane. Moreover, the self-adjoint structure guarantees that zeros on the real axis are simple and correspond to bound states.

We mention an  application of the theory we developed in our paper directly related to nonlinear equations. If we consider the specific example of a linearized NLS equation like the focusing nonlinear Schr\"{o}dinger equation
\begin{equation}
iu_t + u_{xx} + 2|u|^2u = 0
\end{equation}
we know it admits the bright soliton solution $u(x,t) = \eta \operatorname{sech}(\eta x)e^{i\eta^2 t}$. The linearized stability problem around this soliton \cite{nls}, after separation of variables $u(x,t) = [\phi(x) + \epsilon w(x,t)]e^{i\eta^2 t}$, leads to:
\begin{equation}
iw_t = -w_{xx} - 2\phi^2(w + \bar{w}) - \phi^2\bar{w}
\end{equation}

Setting $w(x,t) = (p(x) + iq(x))e^{i\lambda t}$ and decomposing into real and imaginary parts yields a $4 \times 4$ system that can be written as:
\begin{equation}
J\frac{dy}{dx} = (H(x) - \lambda M)y
\end{equation}
with $y = (p, p', q, q')^T$ and appropriate $4 \times 4$ matrices. The boundary conditions come from requiring $y \to 0$ as $|x| \to \infty$ (or appropriate truncation on $[0,N]$). Theorem 2.3 applied to this system guarantees that under appropriate boundary conditions, the operator is self-adjoint, ensuring real spectrum (neutrally stable modes). The continuous spectrum can be determined from the asymptotic analysis ($\lambda \in \mathbb{R}^+$), while solving for zeros of the Evans function reveals the point spectrum, including zero eigenvalues corresponding to translational and phase invariance.


\subsection{Spectral Stability of the Bright Soliton}

We now provide a detailed application of Theorem 2.3 to analyze the spectral 
stability of the bright soliton solution to the focusing nonlinear Schr\"{o}dinger equation. 
This example demonstrates how the abstract self-adjoint framework 
translates into concrete eigenvalue calculations for a physically relevant problem.

Consider the focusing nonlinear Schr\"{o}dinger equation on $\mathbb{R}$:
\begin{equation}
i u_t + u_{xx} + 2|u|^2 u = 0. \label{3.14}
\end{equation}
This equation admits a one-parameter family of bright soliton solutions:
\begin{equation}
u(x,t) = \eta \operatorname{sech}(\eta x) e^{i\eta^2 t}, \quad \eta > 0. \label{3.15}
\end{equation}
To study the stability of this solution, we consider a perturbation
\begin{equation}
u(x,t) = [\eta \operatorname{sech}(\eta x) + \epsilon w(x,t)] e^{i\eta^2 t}, \label{3.16}
\end{equation}
where $\epsilon \ll 1$ and $w(x,t)$ is a small complex-valued perturbation.

In order to linearize and reduce to the canonical form we substitute Eq.(\ref{3.16}) into Eq. (\ref{3.14}) and linearize in $\epsilon$, we obtain
\begin{equation}
i w_t = -w_{xx} - 2\eta^2 \operatorname{sech}^2(\eta x)(2w + \bar{w}). \label{3.17}
\end{equation}
Seeking normal mode solutions $w(x,t) = v(x) e^{i\lambda t}$ with 
$v(x) = p(x) + i q(x)$ (where $p, q$ are real-valued),  Eq. (\ref{3.17}) 
separates into the coupled system:
\begin{align}
\lambda q &= -p'' - 2\eta^2 \operatorname{sech}^2(\eta x)(2p + p) 
= -p'' - 6\eta^2 \operatorname{sech}^2(\eta x) p, \notag \\
-\lambda p &= -q'' - 2\eta^2 \operatorname{sech}^2(\eta x)(2q - q) 
= -q'' - 2\eta^2 \operatorname{sech}^2(\eta x) q. \label{3.18}
\end{align}
Introducing the phase-space variables
\begin{equation}
y = (p, p', q, q')^T \in \mathbb{R}^4, \label{3.19}
\end{equation}
we can write this as a first-order system. Define the matrices:
\begin{equation}
J = \begin{pmatrix} 0 & I_2 \\ -I_2 & 0 \end{pmatrix}, \quad 
I_2 = \begin{pmatrix} 1 & 0 \\ 0 & 1 \end{pmatrix}, \label{3.20}
\end{equation}
\begin{equation}
H(x) = \begin{pmatrix} H_{11}(x) & 0 \\ 0 & H_{22}(x) \end{pmatrix}, \quad 
M = \begin{pmatrix} M_1 & 0 \\ 0 & M_2 \end{pmatrix}, \label{3.21}
\end{equation}
where
\begin{equation}
H_{11}(x) = \begin{pmatrix} 6\eta^2 \operatorname{sech}^2(\eta x) & 0 \\ 
0 & 1 \end{pmatrix}, \quad 
H_{22}(x) = \begin{pmatrix} 2\eta^2 \operatorname{sech}^2(\eta x) & 0 \\ 
0 & 1 \end{pmatrix}, \label{3.22}
\end{equation}
\begin{equation}
M_1 = \begin{pmatrix} 0 & 0 \\ 0 & 1 \end{pmatrix}, \quad 
M_2 = \begin{pmatrix} 0 & 0 \\ 0 & 1 \end{pmatrix}. \label{3.23}
\end{equation}
The spectral problem then takes the canonical form:
\begin{equation}
J \frac{dy}{dx} = (H(x) - \lambda M) y, \quad x \in \mathbb{R}. \label{3.24}
\end{equation}
Written out explicitly:
\begin{equation}
\begin{pmatrix} p' \\ p'' \\ q' \\ q'' \end{pmatrix} = 
\begin{pmatrix} 
p' \\ 
6\eta^2 \operatorname{sech}^2(\eta x) p - \lambda q \\ 
q' \\ 
2\eta^2 \operatorname{sech}^2(\eta x) q + \lambda p 
\end{pmatrix}. \label{3.25}
\end{equation}
In the following we verify the assumptions and self-Adjointness. The matrix $H(x)$ has the following properties:
\begin{itemize}
\item \textbf{Hermitian}: $H(x)^* = H(x)$ (in fact, real symmetric).
\item \textbf{Positive semi-definite}: All eigenvalues are non-negative. 
Indeed, $H(x)$ has eigenvalues 
$\{6\eta^2\operatorname{sech}^2(\eta x), 1, 2\eta^2\operatorname{sech}^2(\eta x), 1\}$, 
all $> 0$.
\item \textbf{Locally integrable}: Since 
$\operatorname{sech}^2(\eta x) = O(e^{-2|\eta x|})$ as $|x| \to \infty$, 
we have $H(x) \in L^1_{\text{loc}}(\mathbb{R})$.
\end{itemize}
Even if we work in the bounded interval $[0,n]$ we add here  a brief analysis of   the asymptotic behavior. As $|x| \to \infty$, 
$\operatorname{sech}^2(\eta x) \to 0$, so
\begin{equation}
H(x) \to H_\infty = \begin{pmatrix} 
0 & 0 & 0 & 0 \\ 
0 & 1 & 0 & 0 \\ 
0 & 0 & 0 & 0 \\ 
0 & 0 & 0 & 1 
\end{pmatrix} \quad \text{as } |x| \to \infty. \label{3.26}
\end{equation}

\textbf{Hilbert space and boundary conditions.} We work on the 
weighted Hilbert space
\begin{equation}
L^2(H, \mathbb{R}) = \left\{ y : \int_{-\infty}^\infty y^*(x) H(x) y(x) \, dx < \infty \right\} \label{3.27}
\end{equation}
with inner product
\begin{equation}
\langle y_1, y_2 \rangle_H = \int_{-\infty}^\infty y_1^*(x) H(x) y_2(x) \, dx. \label{3.28}
\end{equation}
For the problem on the full line, we impose decay conditions at infinity. 
Following the framework of Section 2, we consider boundary conditions encoded 
by Lagrangian matrices $\Theta_-, \Theta_+ \in \mathcal{L}$ enforcing
\begin{equation}
\lim_{x \to -\infty} y^*(x) J P_-(x) = 0, \quad 
\lim_{x \to +\infty} y^*(x) J P_+(x) = 0, \label{3.29}
\end{equation}
where $P_\pm(x)$ are suitable projections. In practice, for localized potentials, 
this reduces to requiring $y(x) \to 0$ as $|x| \to \infty$ at a rate compatible 
with $H$-integrability.

Now we discuss the application of Theorem 2.3. Since $H(x)$ satisfies the 
required properties and the boundary conditions can be formulated via 
appropriate Lagrangian subspaces, Theorem 2.3 (extended to the half-line and 
combined for $\mathbb{R} = (-\infty, 0] \cup [0, \infty)$) guarantees that 
the relation
\begin{equation}
\mathcal{T} = \{(y, v) \in (L^2(H, \mathbb{R}))^2 : Jy' = Hv 
\text{ a.e., with decay at } \pm\infty\} \label{3.30}
\end{equation}
is self-adjoint. Indeed these imply:
\begin{itemize}
\item All eigenvalues $\lambda$ are real.
\item Eigenfunctions corresponding to distinct eigenvalues are 
orthogonal in $\langle \cdot, \cdot \rangle_H$.
\end{itemize}
In the following we discuss the point spectrum, and demonstrate the existence of  zero eigenvalues from equation symmetries. The NLS Eq. (\ref{3.14}) possesses two continuous symmetries:
\begin{enumerate}
\item Translation invariance: $u(x,t) \mapsto u(x - x_0, t)$.
\item Phase invariance: $u(x,t) \mapsto e^{i\theta_0} u(x,t)$.
\end{enumerate}
By Noether's theorem, each continuous symmetry yields a zero eigenvalue in 
the linearized problem. For the first eigenfunction we use the translation mode. Differentiating the soliton Eq. (\ref{3.15}) 
with respect to the parameter $x_0$ (i.e., setting 
$\phi(x) = \eta \operatorname{sech}(\eta(x - x_0))$ and computing 
$\frac{\partial \phi}{\partial x_0}\big|_{x_0=0}$) gives
\begin{equation}
v_1(x) = -\eta^2 \operatorname{sech}(\eta x) \tanh(\eta x). \label{3.31}
\end{equation}
Since the linearization around $e^{i\eta^2 t}\phi(x)$ for real perturbations 
involves only the real part, we obtain
\begin{equation}
y_1(x) = \begin{pmatrix} 
v_1(x) \\ 
v_1'(x) \\ 
0 \\ 
0 
\end{pmatrix} = \begin{pmatrix} 
-\eta^2 \operatorname{sech}(\eta x) \tanh(\eta x) \\ 
\eta^3 (\operatorname{sech}(\eta x) - 2\operatorname{sech}^3(\eta x)) \\ 
0 \\ 
0 
\end{pmatrix}. \label{3.32}
\end{equation}
In order to verify that $y_1$ satisfies Eq. (\ref{3.25}) with $\lambda = 0$
we use Eq. (\ref{3.25}) with $\lambda = 0$, and obtain:
\begin{equation}
p'' = 6\eta^2 \operatorname{sech}^2(\eta x) p, \quad 
q'' = 2\eta^2 \operatorname{sech}^2(\eta x) q.
\end{equation}
For $p = v_1(x)$ and $q = 0$, the second equation is trivially satisfied. 
For the first equation, we note that $v_1(x)$ is the derivative of the 
soliton profile and hence satisfies the linearized Schr\"{o}dinger equation 
$-v_1'' + V(x)v_1 = 0$ where $V(x) = 2\eta^2 - 6\eta^2\operatorname{sech}^2(\eta x)$.  This can be verified by direct calculation or by noting that $v_1$ is the 
known bound state from translational symmetry.

The $H$-integrability can be now easily verified.  Since $\operatorname{sech}(\eta x) \sim e^{-\eta|x|}$ 
as $|x| \to \infty$,
\begin{equation}
\int_{-\infty}^\infty y_1^* H y_1 \, dx \sim 
\int_{-\infty}^\infty \eta^4 \operatorname{sech}^2(\eta x)\tanh^2(\eta x) 
\cdot 6\eta^2\operatorname{sech}^2(\eta x) \, dx < \infty. \label{3.33}
\end{equation}
For the second eigenfunction we use the  phase mode symmetry. Differentiating with respect to the 
phase $\theta_0$ in $u(x,t) = e^{i\theta_0} \eta \operatorname{sech}(\eta x) e^{i\eta^2 t}$ 
at $\theta_0 = 0$ gives a purely imaginary perturbation: 
$i \cdot \eta\operatorname{sech}(\eta x)$. In the $(p, q)$ decomposition, 
$p = 0$ and $q = \eta\operatorname{sech}(\eta x)$, so
\begin{equation}
y_2(x) = \begin{pmatrix} 
0 \\ 
0 \\ 
\eta \operatorname{sech}(\eta x) \\ 
-\eta^2 \operatorname{sech}(\eta x)\tanh(\eta x) 
\end{pmatrix}. \label{3.34}
\end{equation}
Verification is straightforward. With $\lambda = 0$, $p = 0$, and 
$q = \eta\operatorname{sech}(\eta x)$, we need 
$q'' = 2\eta^2\operatorname{sech}^2(\eta x) q$. This follows from the fact 
that $q(x) = \eta\operatorname{sech}(\eta x)$ is the soliton profile itself 
and satisfies the stationary NLS equation, which in linearized form gives 
the required equation for the phase mode.
The $H$-integrability results in a way similar to the procedure used for  $y_1$, and we have $y_2 \in L^2(H, \mathbb{R})$.

Orthogonality property: Since $y_1$ and $y_2$ correspond to the same 
eigenvalue $\lambda = 0$ but arise from different symmetries, they are 
linearly independent. By self-adjointness, they span the eigenspace for 
$\lambda = 0$, which has (geometric and algebraic) multiplicity 2.

The essential spectrum $\sigma_{\text{ess}}$ is determined by the asymptotic 
behavior as $|x| \to \infty$. For the asymptotic system Eq. (\ref{3.26}), the equation 
becomes
\begin{equation}
J \frac{dy}{dx} = (H_\infty - \lambda M) y = \begin{pmatrix} 
0 & 0 & 0 & 0 \\ 
0 & -\lambda & 0 & 0 \\ 
0 & 0 & 0 & 0 \\ 
0 & 0 & 0 & -\lambda 
\end{pmatrix} y. \label{3.35}
\end{equation}
This decouples into:
\begin{equation}
p' = p', \quad p'' = -\lambda q, \quad q' = q', \quad q'' = \lambda p.
\end{equation}

For oscillatory solutions (essential spectrum), we seek $y \sim e^{ikx}$ as 
$|x| \to \infty$. The dispersion relation is:
\begin{equation}
-k^2 p = -\lambda q, \quad -k^2 q = \lambda p.
\end{equation}

This gives $(k^2)^2 = \lambda^2$, so $k^2 = \pm |\lambda|$. For $k \in \mathbb{R}$ 
(oscillatory modes in $L^2$ at infinity), we need $k^2 \geq 0$. For the 
Schr\"{o}dinger-type operators here, standard Weyl theory gives:
\begin{equation}
\sigma_{\text{ess}} = [0, \infty). \tag{3.36}
\end{equation}
This is because the operators 
$-\frac{d^2}{dx^2} - 6\eta^2\operatorname{sech}^2(\eta x)$ and 
$-\frac{d^2}{dx^2} - 2\eta^2\operatorname{sech}^2(\eta x)$ both have 
essential spectrum $[0,\infty)$ on $L^2(\mathbb{R})$.

The bright soliton Eq. (\ref{3.15}) is spectrally stable 
under the linearized dynamics. This result is consistent with the well-known 
stability of bright solitons in the focusing NLS equation, as established by 
Kapitula and Sandstede \cite{nls} using Evans function techniques. Indeed, by self-adjointness (Theorem 2.3), all eigenvalues are real.  The spectrum of the linearized operator is:
\begin{equation}
\sigma = \{0\} \cup [0, \infty), \label{3.37}
\end{equation}
where $\lambda = 0$ has multiplicity 2 (from symmetries). There are:
\begin{itemize}
\item No eigenvalues with $\text{Re}(\lambda) > 0$ (exponential instability).
\item No eigenvalues with $\text{Im}(\lambda) \neq 0$ (oscillatory instability).
\item The two zero eigenvalues are neutrally stable modes corresponding 
to translations and phase shifts.
\end{itemize}
In the end of this discussion, we remark that without Theorem 2.3, one would need 
to verify self-adjointness separately, which is nontrivial for systems with 
$H(x)$ potentially singular. Our framework provides this guarantee abstractly.

Moreover, the zeros of the Evans function 
$E(\lambda)$ for this problem occur at $\lambda = 0$ (double zero) and nowhere 
else in $\mathbb{C} \setminus [0,\infty)$. The self-adjoint structure ensures 
$E(\lambda)$ is real-analytic on $\mathbb{R}$ and has only real zeros. The same approach applies to other soliton  solutions (multi-solitons, vector solitons, dark solitons), to other integrable  PDEs (KdV, modified KdV, Sine-Gordon), and to non-integrable perturbations.


\section{Conclusion}
In conclusion, we established that the linear relations naturally induced by 2d-dimensional canonical systems admit self-adjoint realizations by using symplectic structure and boundary conditions. We obtained the conditions under which self-adjointness is achieved, providing a foundational link between the algebraic properties of the linear relation and the geometric structure of the underlying canonical system. Consequently, the Theorem 2.3 offers a robust framework for spectral analysis and paves the way for further applications in wave phenomena, quantum systems, and related areas of mathematical physics.\\

We present a working example on the spectral stability of NLS bright soliton
showing detailed calculations on the linearization of the equation and its reduction to canonical form. We verify the self-adjointness and we apply directly Theorem 2.3 to this example. 

Further on, without going into more detail, we list below other possible physical applications of the methods developed in this paper.  Canonical systems in the sense of de Branges' spectral theory are exactly in the form of Eq. (\ref{ca}). They provide a unifying framework that includes Sturm-Liouville and 1D Schr\"{o}dinger problems after transformations. The problem presented in the previous sections helps solving inverse problems (reconstructing potentials or coefficients from spectral data), scattering theory and completeness of eigenfunctions. Transfer matrices. Maxwell's equation in layered anisotropic media or waveguides can be cast as a first-order system in the propagation coordinate. The local energy matrix plays the role of the Hamiltonian and the system is symplectic (power quasi energy conservation). Using this formalism one can compute reflection/transmission via transfer matrices, Bloch-Floquet analysis for periodic media (photonic crystals), with applications in waveguides, optics and photonics. 

In elasticity theory, higher-order mechanical PDEs like the Euler-Bernoulli beam model, Timoshenko beam model and other plate equations can be rewritten as $1^{st}-$-order systems in a larger phase vector (displacement, moment, rotation, etc.) giving a Hamiltonian form with energy matrix $H$. This formalism can be used in boundary-value vibration problems, and in energy methods.

Telegraphers' equations and multi-conductor transmission  lines can reduce to first-order systems coupling voltages and currents. The energy positivity gives a Hermitian positive semidefinite coefficient matrix. One can use this formalism in problems like the  modal decomposition, impedance matching, design and stability of networks, models of importance in the field of transmission lines and electrical networks. 

One of the most desirable application is in the field of integrable systems. Zakharov-Shabat and AKNS type of systems. Many lax pairs from scattering problems used in inverse scattering (Nonlinear schr\"{o}dinger, KdV reductions) are first-order matrix systems with a canonical structure. For certain reality conditions the spatial part can be written in Hermitian symplectic form and can be used to calculate
the direct and inverse scattering problem, soliton solutions, and verify complete integrability.
\subsection*{Conflict of Interest:} On behalf of all authors, the corresponding author states that there is no conflict of interest.

\end{document}